\documentclass{amsart}

\usepackage{graphicx}
\usepackage{amsmath,latexsym}
\usepackage{amsfonts,amssymb}
\usepackage{amscd, amsthm}
\usepackage{stmaryrd}
\usepackage{fancyhdr}
\usepackage{commath,enumitem,chngcntr}
\usepackage{indentfirst, hyperref}
\usepackage[titletoc,toc,title]{appendix}
\usepackage{bbm}
\usepackage{pst-node}
\usepackage{tikz-cd} 
\usepackage{pgf,tikz}
\usepackage{mathrsfs}
\usetikzlibrary{arrows,calc,positioning}
\usepackage{float}
\usepackage{subcaption}

\numberwithin{equation}{section}

\newcommand{\rvline}{\hspace*{-\arraycolsep}\vline\hspace*{-\arraycolsep}}

\newtheorem{theorem}{Theorem}[section]
\newtheorem{lemma}[theorem]{Lemma}

\newtheorem{conjecture}[theorem]{Conjecture}

\begin{document}
\title{A note on graph drawings with star-shaped boundaries in the plane}

\author{Yanwen Luo}
\address{Department of Mathematics, University of California, Davis
, California 95616}
\email{ywluo@ucdavis.edu}
\thanks{The author was supported in part by NSF Grant DMS-1719582.}


\begin{abstract}
In this short note, we propose a straightforward method to produce an straight-line embedding of a planar graph where one face of a graph is fixed in the plane as a star-shaped polygon. It is based on minimizing discrete Dirichlet energies, following the idea of Tutte's embedding theorem. We will call it a geodesic triangulation of the star-shaped polygon. Moreover, we study the homotopy property of spaces of all straight-line embeddings. We give a simple argument to show that this space is contractible if the boundary is a non-convex quadrilateral. We conjecture that the same statement holds for general star-shaped polygons. 

\end{abstract}

\maketitle

\section{Introduction}

A triangulation of a fixed combinatorial type of $T$ on a surface with a Riemannian metric $(S, g)$ is a \textit{geodesic triangulation} if each edge in $T$ is embedded as a geodesic arc in $S$. We study the following two problems concerning the space of geodesic triangulations with a fixed combinatorial type on a polygon in the Euclidean plane. 
\begin{enumerate}
	\item The \textit{embeddability} problem: Given a surface $(S, g)$ with a triangulation $T$,  can we construct a geodesic triangulation with the combinatorial type of $T$? In particular, if $S$ is a 2-disk with a triangulation $T$ and we specify the positions of the boundary vertices of $T$ in the plane so that they form a polygon, can we find positions of the interior vertices in the plane to construct a geodesic triangulation of $S$ with the combinatorial type of $T$?
	\item The \textit{contractibility} problem: If the space of geodesic triangulations on $(S, g)$ with a fixed combinatorial type of $T$ is not empty, what is the topology of this space? In particular, is it a contractible space?
\end{enumerate}

Tutte's embedding theorem is one of the most fundamental results in graph theory. It provides a simple answer to the embeddability problem by fixing a face of the $1$-skeleton of a triangulation $T$ of a disk, considered as a planar graph, to be a convex polygon and solving a linear system to find the locations of other vertices. This idea has been generalized to construct embeddings of graphs on general surfaces \cite{de1991comment, hass2012simplicial}.

Tutte's construction fails if the boundary is not convex. Various methods have been proposed to construct embeddings of graphs given a general boundary condition by Hong-Nagamochi \cite{hong2008convex} and Xu et al.\cite{xu2011embedding}. In this note, we give a constructive method to produce geodesic triangulations with a fixed combinatorial type for a star-shaped polygon.

The contractibility problem has been studied in \cite{bing1978linear, cairns1944deformations, cairns1944isotopic, ho1973certain, 1910.00240}, originally motivated to solve problems of determining the existence and uniqueness of differentiable structures on triangulated manifolds \cite{connelly1983problems}. From the viewpoint of a novel field of discrete differential geometry, the spaces of geodesic triangulations of a fixed homotopy type can be regarded as discrete versions of the group of diffeomorphisms of the 2-disk fixing the boundary pointwise. A classical theorem by Smale \cite{smale1959diffeomorphisms} states that this group is contractible.  Bloch-Connelly-Henderson proved that the space of geodesic triangulations of a fixed combinatorial type $T$ of a convex polygon $\Omega$, denoted as $\mathcal{GT}(\Omega, T)$, is homeormophic to some Euclidean space $\mathbb{R}^k$. This result discretizes Smale's theorem. 

The results about homotopy types of spaces of geodesic triangulations of general polygons are sporadic. Using an induction argument, Bing-Starbird \cite{bing1978linear} showed that if $\Omega$ is a star-shaped polygon with a triangulation $T$ in the plane, and $T$ does not contain any edge connecting two boundary vertices, then $\mathcal{GT}(\Omega, T)$ is non-empty and path-connected.
They constructed an example showing that $\mathcal{GT}(\Omega, T)$ was not necessarily path-connected if we didn't assume star-shaped boundary. More complicated examples given in  \cite{luo2022spaces} show that the spaces of geodesic triangulations of general polygons can be complicated, in the sense that their homotopy groups can have large ranks. We answer the contractibility problem in a special case where the boundary polygon is a non-convex quadrilateral. 

In Section 2, we recall Tutte's theorem and its  generalizations. In Section 3, we give an explicit construction of a geodesic triangulation in $\mathcal{GT}(\Omega, T)$ if $\Omega$ is a strictly star-shaped polygon. In Section 4, we show that the space of geodesic triangulations of a star-shaped quadrilateral is contractible.

\section{Tutte's embedding and its generalization}
\subsection{Tutte's embedding for the disk}

Given a triangulation $T = (V, E, F)$ of the 2-disk with the sets of vertices $V$, edges $E$ and faces $F$,  the $1$-skeleton of $T$ is a planar graph. There is no canonical method to embed this graph in the plane. Tutte\cite{tutte1963draw} provided an efficient method to construct a straight-line embedding of a 3-vertex-connected planar graph by specifiying the coordinates of vertices of one face as a convex polygon and solving for the coordinates of other vertices with a linear system of equations. Using a discrete maximal principle, Floater\cite{floater2003one} proved the same result for triangulations of the 2-disk. Gortler, Gotsman, and Thurston\cite{gortler2006discrete} reproved Tutte's theorem with \textit{discrete one forms} and generalized this results to the case of multiple-connected polygonal regions with appropriate assumptions on the boundaries. Since we are dealing with triangulations, we use the formulation given by Floater\cite{floater2003one}.

\begin{figure}[h!]
  \includegraphics[width=0.65\linewidth]{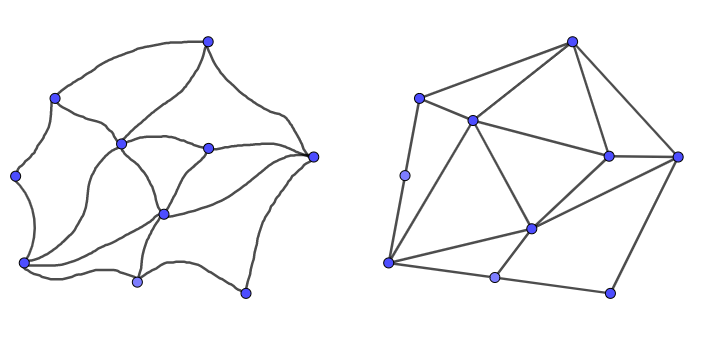}
  \caption{Tutte's embedding}
\end{figure}
\begin{theorem}
	Assume $T = (V, E, F)$ is a triangulation of a convex polygon $\Omega$, and $\phi$ is a simplexwise linear homeomorphism from $T$ to $\mathbb{R}^2$. If  $\phi$ maps every interior vertex in $T$ into the convex hull of the images of its neighbors, and maps the cyclically ordered boundary vertices of $T$ to the cyclically ordered boundary vertices of $\Omega$, then $\phi$ is one to one.
\end{theorem}
This theorem gave a discrete version of the Rado-Kneser-Choquet theorem about harmonic maps from the disk to a convex polygon. Moreover, it gives a constructive method to produce geodesic triangulations of a convex polygon with the combinatorial type of $T$.

First assign a positive weight $c_{ij}$ to a directed edge $(i, j) \in \bar{E}$, where $\bar{E}$ is the set of directed edges of $T$. We normalize the weights by 
$$w_{ij} = \frac{c_{ij}}{\sum_{j\in N(v_i)} c_{ij}}$$
where the set $N(v_i)$ consists of all the vertices that are neighbors of $v_i$, so that $\Sigma_{j\in N(v_i)}w_{ij} = 1$ for all $i = 1, 2,..., N_I$. Notice that we don't impose symmetry condition $w_{ij} = w_{ji}$. We are given the coordinates $\{(b_i^x, b_i^y)\}_{i = N_I + 1}^{|V|}$ for all the boundary vertices such that they form a convex polygon $\Omega$ in $\mathbb{R}^2$. Then we can solve the following linear system
$$\sum_{j\in N(v_i)}w_{ij}x_j = x_i \quad i = 1, 2, ... N_I;$$
$$\sum_{j\in N(v_i)}w_{ij}y_j = y_i \quad i = 1, 2, ... N_I;$$
$$x_i = b_i^x \quad i = N_I +1, N_I + 2, ... N_I + N_B = |V|;$$
$$y_i = b_i^y  \quad i = N_I +1, N_I + 2, ... N_I + N_B = |V|$$
where $N_I = |V_I|$ is the size of the set of interior vertices $V_I$, and $N_B = |V_B|$ is the size of the set of boundary vertices $V_B$. The solution to this linear system produces the coordinates of all the interior vertices in $\mathbb{R}^2$. We put the vertices in the positions given by their coordinates, and connect the vertices based on the combinatorics of the triangulation $T$. Tutte's theorem claims that the result is a geodesic triangulation of $\Omega$ with the combinatorial type of $T$.

The linear system above implies that the $x$-coordinate(or $y$-coordinate) of one interior vertex is a convex combination of the $x$-coordinates(or $y$-coordinates) of its neighbors.  Notice that the coefficient matrix of this system is not necessarily symmetric but it is diagonally dominant, so the solution exists uniquely.

 Tutte's theorem solves the embeddability  problem for a triangulation of a convex polygon. We can vary the coefficients $w_{ij}$ to construct families of geodesic triangulations of a convex polygon. This idea will lead to a simple proof of the contractibility of the space of geodesic triangulations of convex polygons \cite{luo2022spaces}.

\section{Drawing graphs with star-shaped Boundary}

In this section, we consider a star-shaped subset $\Omega$ of $\mathbb{R}^2$. An \textit{eye} of a star-shaped region $\Omega$ is a point $p$ in $\Omega$ such that for any other point $q$ in $\Omega$ the line segment $l(t) = tp + (1-t)q$ lies inside $\Omega$. The set of eyes of $\Omega$ is called the \textit{kernel} of $\Omega$. A set is called \textit{strictly star-shaped} if the interior of the kernel is not empty. 

In the case of polygons in $\mathbb{R}^2$, the kernel is the intersection of a family of closed half-spaces, each defined by the line passing one boundary edge of $\Omega$. Every closed half space contains a half disk in $\Omega$ centered at one point on its corresponding boundary edge. If the star-shaped polygon is strict, the intersection of the \textit{open} half-spaces is not empty. This means that we can pick an eye $e$ with a neighborhood $U$ of $e$ such that if $q\in U$, then $q$ is also an eye of $\Omega$.

The first question to address is how to construct a geodesic triangulation of a strictly star-shaped polygon $\Omega$ with a combinatorial type of $T$. As Bing and Starbird\cite{bing1978linear} pointed out, it was not always possible if there was a  edge connecting boundary vertices, called \emph{dividing edge}.  Assuming there was no dividing edge in $T$, they proved that such geodesic triangulations existed by induction. 

We give an explicit method to produce a geodesic triangulation for a strictly star-shaped polygon. We can regard all the edges $e_{ij}$ in $T$ as ideal springs with Hook constants $w_{ij}$. Fixing the boundary vertices, the equilibrium state corresponds to the critical point of the \textit{weighted length energy} defined as 
$$\mathcal{E} = \frac{1}{2}\sum_{e_{ij}\in E_I} w_{ij}L_{ij}^2$$ 
where $L_{ij}$ is the length of the edge connecting $v_i$ and $v_j$.  This energy can be regarded as a discrete version of the Dirichlet energy \cite{de1991comment, hass2012simplicial}, and it has a unique minimizer corresponding to the equilibrium state. Tutte's theorem guarantees that the equilibrium state is a geodesic embedding of $T$ if the boundary is a convex polygon.

Given a triangulation $T$ of a fixed strictly star-shaped polygon $\Omega$, assume that  the weighted length energy $\mathcal{E} $ satisfies $\sum_{e_{ij}\in E_I} w_{ij} = 1$. Notice that if the polygon is star-shaped but not convex, we can't choose arbitrary weights to generate a geodesic embedding of $T$. Hence we need to assign weights carefully to avoid singularities such as intersections of edges and degenerate triangles.

The idea is to distribute more and more weights to the interior edges connecting two interior vertices. As the weights for interior edges connecting two interior vertices tend to $1$, all the interior vertices will concentrate at a certain point. If we can choose this point to be an eye of the polygon, we will produce an geodesic embedding of $T$ of $\Omega$. 

Fix a polygon $\Omega$ with a triangulation $T$ and the coordinates $\{(b_j^x, b_j^y)\}_{i = N_I+1}^{|V|}$ for its boundary vertices. Given a set of coordinates in $\mathbb{R}^2$ for all the interior vertices $\{(x_i, y_i)\}_{i = 1}^{N_I}$, we define a family of weighted length energies with a parameter $0<\epsilon<1$ as
$$\mathcal{E}(\epsilon) = \frac{1 - \epsilon}{2M_I} \sum_{e_{ij}\in E_I^I} L_{ij}^2 + \frac{\epsilon}{2M_B} \sum_{e_{ij}\in E_I^B} L_{ij}^2$$
where $E_I^B$ is the set of all the interior edges connecting an interior vertex to a boundary vertex and $E_I^I$ is the set of all the interior edges connecting two interior vertices. Let $M_B = |E_I^B|$ and $M_I = |E_I^I|$. The edge lengths $L_{ij}$ are determined by the coordinates of the vertices
$$L_{ij}^2 = (x_i - x_j)^2 + (y_i - y_j)^2.$$

As $\epsilon \to 0$, most weights are assigned to interior edges in $E_I^I$, forcing all the interior vertices of the minimizer of $\mathcal{E}(\epsilon)$ to concentrate to one point. 
\begin{theorem}
Let $\Omega$ be a polygonal region with a triangulation $T$ of $\Omega$. Let $v^B_j = (x^B_j, y^B_j) = (b_j^x, b_j^y)$ for $j = 1, ..., N_B$ be the coordinates of the boundary vertices of $\Omega$ and $v^I_i(\epsilon) = (x^I_i(\epsilon), y^I_i(\epsilon))$ for $i = 1, ..., N_I$ be the coordinates of the interior vertices of the minimizer of the energy $\mathcal{E}(\epsilon)$. Then for all $i = 1, 2, ...., N_I$, 
$$\lim_{\epsilon \to 0}v^I_i = \lim_{\epsilon \to 0}(x^I_i(\epsilon), y^I_i(\epsilon)) = (x_0, y_0) = v_0$$
where  
$$v_0 =  \sum_{j =1}^{ N_B}\lambda_j v^B_j \quad \text{and} \quad \lambda_j = \frac{deg(v^B_j) - 2}{\sum_j deg(v^B_j - 2)} = \frac{deg(v^B_j) - 2}{M_B},$$
assuming $deg(v)$ is the degree of the vertex $v$ in $T$.
\end{theorem}
\begin{proof}
The minimizer of $\mathcal{E}(\epsilon)$ satisfies the following linear system formed by taking derivatives with respect to $x_i$ and $y_i$ for all $i = 1, 2, ...., N_I$
$$\frac{1 - \epsilon}{M_I}\sum_{i\in N(v^I_k)} (v^I_k - v^I_i) + \frac{\epsilon}{M_B}\sum_{j\in N(v^I_k)} (v^I_k - v^B_j) = 0 \quad \text{for } k = 1, 2, \cdots N_I.$$
Notice that we separate the interior vertices $v_i^I \in V_I$ and the boundary vertices $v_j^B \in V_B$ in the summation. This system can be represented as 
$$M(\epsilon)x = b_x \quad M(\epsilon)y = b_y$$
where the variables are $$x = (x^I_1, x^I_2,..., x^I_{N_I}, x^B_1, ..., x^B_{N_B})^T$$ and $$y = (y^I_1, y^I_2,..., y^I_{N_I}, y^B_1, ..., y^B_{N_B})^T.$$ The boundary conditions are $$b_x = (0, 0, ..., 0, x^B_1, ..., x^B_{N_B})^T$$ and $$b_y = (0, 0, ..., 0, y^B_1, ..., y^B_{N_B})^T.$$
The coefficient matrix $M(\epsilon)$ is an $(N_I + N_B)\times(N_I + N_B)$ matrix, and it can be decomposed as
$$
M(\epsilon) = \begin{pmatrix}
  S(\epsilon)
  & \rvline & -\epsilon W \\
\hline
  0 & \rvline &
  Id
\end{pmatrix}
$$
where $W$ is an $N_I\times N_B$ matrix, $S(\epsilon)$ is a square matrix of size $N_I$, and $Id$ is the identity matrix of size $N_B$. 
The matrix $W$ is defined as
\[ W(i, j) = \left\{ \begin{array}{ll}
         \frac{1}{M_B} & \mbox{if $v_i^I$ is connected to $v_j^B$};\\
        0 & \mbox{if $v_i^I$ is not connected to $v_j^B$}.\end{array} \right. \] 
 The matrix $S$ is defined as
\[ S(i, j)(\epsilon) = \left\{ \begin{array}{ll}
          -\sum_{i\neq k} S(i, k) + \epsilon\sum_{k = 1}^{N_B} W(i, k)& \mbox{if $i = j$};\\
         -\frac{1 - \epsilon}{M_I} & \mbox{if $v_i^I$ is connected to $v_j^I$};\\
        0 & \mbox{if $v_i^I$ is not connected to $v_j^I$}.\end{array} \right. \] 

Notice that for the first $N_I$ rows in $M(\epsilon)$, the sums of their respective entries  are zero, and all the off-diagonal terms are non-positive. The matrix $W$ represents the relations of the boundary vertices with the interior vertices, and the sum of all its entries equals one. The matrix $S(\epsilon)$ is symmetric, strictly diagonally-dominant, and the sum of all its entries equals $\epsilon$. 

To show the limiting behavior of the solution to the system as $\epsilon \to 0$, we need the lemma below.
\begin{lemma}
Given the notations above, we have 
$$\lim_{\epsilon \to 0} \epsilon S(\epsilon)^{-1} = \mathbbm{1}$$
where the matrix $\mathbbm{1}$ is the $N_I\times N_I$ matrix with all entries equal to $1$.
\end{lemma}
\begin{proof}
Notice that $S(\epsilon)$ is symmetric and strictly diagonally dominant, so it is invertible. Let $S = S(0)$ and $M = M(0)$, then $S$ has an eigenvalue $\lambda = 0$ with the normalized eigenvector $v = (1/\sqrt{N_I}, 1/\sqrt{N_I}, ..., 1/\sqrt{N_I})^T$. 

First, we show that $\lambda = 0$ is a simple eigenvalue for $S$. If $S$ has another eigenvector $u = (u_1, u_2, ..., u_{N_I})^T$ corresponding to $\lambda = 0$ not parallel to $v$, then it is orthogonal to $v$ so $\sum_i u_i = 0$. Without loss of generality, we assume that $u_1>0$ achieves the maximal absolute value among $u_i$. Then we have 
$$Su = 0 \quad \Rightarrow \quad \sum_{i = 1}^{N_I} S(1, i)u_i = 0 \quad \Rightarrow \quad  S(1, 1)u_1 = - \sum_{i = 2}^{N_I} S(1, i)u_i.$$
Notice that $S$ is weakly diagonally dominant,  $S(1, 1)> 0$, and  $S(1, i) \leq 0$, so we can deduce that
$$S(1, 1)u_1 \geq - \sum_{i = 2}^{N_I}S(1, i)u_1 \quad \Rightarrow \quad - \sum_{i = 2}^{N_I}S(1, i)(u_i - u_1) \geq 0.$$
By our assumption, $u_i - u_1 \leq 0$ for all $i = 1, ..., N_I$, so the only possibility is $u_i = u_1$ for all $i$, which contradicts to the fact that $u$ is orthogonal to $v$. Hence all the other eigenvalues of $S$ are positive by Gershgorin circle theorem. (See, e.g. \cite{golub2013matrix})

Second, we show that the eigenvalue $\lambda(\epsilon)$ of $S(\epsilon)$ approaching to $0$ satisfies 
$$\lim_{\epsilon\to 0}\frac{\lambda(\epsilon)}{\epsilon} = \frac{1}{N_I}.$$
This means that the derivative $(d\lambda/d\epsilon)(0) = 1/N_I$.
To compute the derivative, notice that the sum of all the entries of $S(\epsilon)$ is $\epsilon$, hence we have

\begin{align*}
v^TS(\epsilon)v = \frac{1}{N_I}
    (1, 1, ..., 1)S(\epsilon) & \begin{pmatrix}
           1 \\
           1 \\
           \vdots \\
          1
         \end{pmatrix} = \frac{\epsilon}{N_I}.
  \end{align*}
The derivative of a simple eigenvalue of a symmetric matrix is given in  \cite{petersen2008matrix} by 
$$\frac{d\lambda}{d\epsilon}(0)= \frac{d(v^TS(\epsilon)v)}{d\epsilon}  = \frac{d(\epsilon/N_I)}{\epsilon} = \frac{1}{N_I}.$$
  
  Finally, we are ready to prove the lemma.  Since $S(\epsilon)$ is symmetric, we have the diagonalization with an orthonormal matrix $P(\epsilon)$
 $$\epsilon S^{-1}(\epsilon) = P(\epsilon)\begin{pmatrix}
           \epsilon\lambda_1^{-1}(\epsilon) \\
           &\epsilon \lambda_{2}^{-1}(\epsilon)  \\
           & & \ddots \\
          & & & \epsilon \lambda_{N_I}^{-1}(\epsilon) 
         \end{pmatrix}P^T(\epsilon).$$
Without loss of generality, we assume the first eigenvalue  $\lim_{\epsilon\to 0}\lambda_1(\epsilon) = 0$. Given any $0<\delta<1$, we can choose small $\epsilon>0$ such that  the following three inequality holds
 $$\lambda_i(\epsilon) > C > 0 \text{ for } i = 2, 3, ..., N_i;$$
  $$||P(\epsilon)\begin{pmatrix}
           \epsilon\lambda_1^{-1}(\epsilon) \\
           &\epsilon \lambda_{2}^{-1}(\epsilon)  \\
           & & \ddots \\
          & & & \epsilon \lambda_{N_I}^{-1}(\epsilon) 
         \end{pmatrix}P^T(\epsilon) - P(\epsilon)\begin{pmatrix}
           N_I \\
           &0 \\
           & & \ddots \\
          & & & 0
         \end{pmatrix}P^T(\epsilon)||_2< \delta;$$
         and the eigenvector $v_1(\epsilon)$ of $S(\epsilon)$ corresponding to the eigenvector $\lambda_1(\epsilon)$ satisfies $$||v_1(\epsilon) - \frac{1}{\sqrt{N_I}}\begin{pmatrix}
           1 \\
           1 \\
           \vdots \\
          1
         \end{pmatrix}||_\infty < \delta.$$ 
        Notice that the columns of $P(\epsilon) = (v_1, v_2, ..., v_{N_I})$ form a set of the orthonormal basis formed by eigenvectors $v_i$, where the first eigenvector $v_1(\epsilon)$ approaches $v = (1/\sqrt{N_I}, ..., 1/\sqrt{N_I})$. Then we have
         \begin{align*}
         ||\epsilon S^{-1}(\epsilon) - \mathbbm{1}||_2 \leq & ||P(\epsilon)\begin{pmatrix}
           \epsilon\lambda_1^{-1}(\epsilon) \\
            & \ddots \\
            & &\epsilon \lambda_{N_I}^{-1}(\epsilon) 
         \end{pmatrix}P^T(\epsilon) - P(\epsilon)\begin{pmatrix}
           N_I \\
            & \ddots \\
          & & 0
         \end{pmatrix}P^T(\epsilon)||_2  \\
         + &   || P(\epsilon)\begin{pmatrix}
           N_I \\
            &\ddots \\
           & & 0
         \end{pmatrix}P^T(\epsilon) - \mathbbm{1}||_2 \leq \delta + || N_Iv_1^T(\epsilon)v_1(\epsilon) - \mathbbm{1}||_2. \\
         \end{align*}
        Notice that 
        $$|| N_Iv_1^T(\epsilon)v_1(\epsilon) - \mathbbm{1}||_2 \leq 2N_I^2\delta.$$
        Hence
        $$||\epsilon S^{-1}(\epsilon) - \mathbbm{1}||_2 \leq (1 + 2N_I^2)\delta.$$
        
  \end{proof}
  
The inverse of the matrix $M(\epsilon)$ can be represented as
$$
M^{-1}(\epsilon) = \begin{pmatrix}
  S^{-1}(\epsilon)
  & \rvline & \epsilon S^{-1}(\epsilon) W \\
\hline
  0 & \rvline &
  I
\end{pmatrix}.
$$
Then the solution of the linear system $M(\epsilon)x = b_x$ is $x = M^{-1}(\epsilon)b_x$, whose first $N_I$ entries are given by
$$\begin{pmatrix}
           x^I_1(\epsilon) \\
           x^I_2(\epsilon) \\
           \vdots \\
          x^I_{N_I}(\epsilon)
         \end{pmatrix} = \epsilon S^{-1}(\epsilon)W\begin{pmatrix}
           x^B_1\\
           x^B_2 \\
           \vdots \\
          x^B_{N_B}
         \end{pmatrix}. $$
         As $\epsilon \to 0$, the solution approaches $\mathbbm{1}Wx^B$. All the $x_i^I$ approach the same point
         $$\lim_{\epsilon \to 0}x_i^I =  (1, ..., 1)W \begin{pmatrix}
           x^B_1\\
           x^B_2 \\
           \vdots \\
          x^B_{N_B}
         \end{pmatrix}  = \sum_{i  = 1}^{N_B} \frac{deg(v_i^B) - 2}{N_B}x_i^B.$$
         A similar result holds for $y$-coordinates of the interior vertices. Hence we conclude the limit of the solutions $\lim_{\epsilon \to 0}v^I_i = v_0$.
\end{proof}
Notice that the matrix $W$ can be replaced with more general matrices. The original energy $\mathcal{E}(\epsilon)$ distributes $\epsilon$ percentage of weights evenly to all the edges in $E^B_I$. We can define new energies by redistributing the weights

$$\mathcal{E}^W(\epsilon) = \frac{1 - \epsilon}{2M_I} \sum_{e_{ij}\in E_I^I}L_{ij}^2 + \frac{\epsilon}{2} \sum_{e_{ij}\in E_I^B}w_{ij} L_{ij}^2$$
with $w_{ij}>0$ and $\sum_{(i, j)\in E_I^B}w_{ij} = 1$. The matrix $W$ is defined as
\[ W(i, j) = \left\{ \begin{array}{ll}
         w_{ij} & \mbox{if $v_i^I$ is connected to $v_j^B$};\\
        0 & \mbox{if $v_i^I$ is not connected to $v_j^B$}.\end{array} \right. \] 
  The limit of the solution is
  $$v_0 =  \sum_{j = 1}^{N_B}\lambda_j v^B_j \quad \text{ where }\quad \lambda_j = \sum_{i = 1}^{N_I} w_{ij}.$$

  To construct a geodesic triangulation, pick an eye $e$ of $\Omega$ such that $e = \sum_{i  =1}^{N_B} \lambda_i v_i^B$ where $\lambda_i >0$ and  $\sum_{i  =1}^{N_B} \lambda_i = 1$, then define 
   \[ W(i, j) = \left\{ \begin{array}{ll}
         w_{ij}  = \frac{\lambda_i}{deg(v_j^B) - 2}& \mbox{if $v_i^I$ is connected to $v_j^B$};\\
        0 & \mbox{if $v_i^I$ is not connected to $v_j^B$}.\end{array} \right. \] 
        and the corresponding energy $\mathcal{E}^W(\epsilon)$. The remaining task is to show that the critical point of $\mathcal{E}^W(\epsilon)$ is a geodesic embedding of $T$ for small $\epsilon$.
        
  If $\Omega$ is not convex, there exists a \textit{reflex vertice}, defined as a boundary vertice of $\Omega$ where the turning angle is negative. We use the result by Gortler, Gotsman and Thurston\cite{gortler2006discrete} to show that the minimizer of $\mathcal{E}^W(\epsilon)$ constructed above is an embedding for some $\epsilon>0$.
\begin{theorem}[\cite{gortler2006discrete}]
Given a strictly star-shaped polygon $\Omega$ with a triangulation $T$ without dividing edges, if the reflex vertices of $\Omega$ are in the convex hull of their respective neighbors, then the solution to the linear system generates a straight-line embedding of $T$.
\end{theorem}

  \begin{theorem}
   Given a strictly star-shaped polygon $\Omega$ with a triangulation $T$ without dividing edges, and an eye $e$ in $\Omega$ with coefficients $W$,  there exists an $\epsilon>0$ such that the critical point of the energy $\mathcal{E}^W(\epsilon)$ generates a geodesic embedding of $T$. 
  \end{theorem}
  \begin{proof}
  Theorem 4.3 implies that we only need to check that the reflex vertices $v_r$ are in the convex hulls of their respective neighbors. 
  
  Choose an $\epsilon$ small enough such that the vertices of the critical point of $\mathcal{E}^W(\epsilon)$ defined above are eyes of $\Omega$. Assume $v_r$ is a reflexive point on the boundary of $\Omega$.  Let $v$ be an interior vertex of the geodesic triangulation in the star of $v_r$, and let $v_1$ and $v_2$ be the two boundary vertices connecting to $v_r$. Since there is no dividing edge in $T$, $v_1$ and $v_2$ are the only boundary vertices connecting to $v_r$. We want to show that $v_r$ is in the convex hull of its neighbors. 
  
 Assume the opposite, then all the edges connecting to $v_r$ lie in a closed half plane, so the inner product of any pair of three vectors $\overrightarrow{v_rv_1}$, $\overrightarrow{v_rv_2} $ and $\overrightarrow{v_rv}$ is non-negative. But the inner angle at $v_r$ is larger than $\pi$, then either angle $\angle v_1v_rv$ or $\angle vv_rv_2$ is strictly larger than than $\frac{\pi}{2}$, which means one inner product is negative. This leads to a contradiction.
  \end{proof}
  This result solves the embeddability problem for strictly star-shaped polygons $\Omega$ with a triangulation $T$. We can construct a geodesic triangulation of $\Omega$ as follows. Pick an eye $e$ of $\Omega$ with the coefficients $W$ defined above. Then choose $\epsilon = 1/2$ and solve the linear system corresponding to the critical point of $\mathcal{E}^W(1/2)$. If the solution is not an embedding, replace $\epsilon$  by $\epsilon/2$ and continue. 
  
  One interesting question is whether we can realize any given geodesic triangulation in $\mathcal{GT}(\Omega, T)$ as the critical point of certain weighted length energy by choosing appropriate weights. Unfortunately, this is not the case. The example in Eades, Healy, and Nikolov \cite{eades2018weighted} below shows that not every geodesic triangulation of a disk can be realized as the minimizer of certain energy. It closely related to Schönhardt polyhedron.

  \begin{figure}[h!]
 \includegraphics[width=0.4\linewidth]{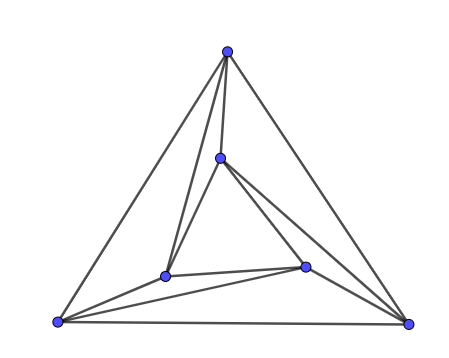}
 \caption{Example in \cite{eades2018weighted} which is not a minimizer.}
\end{figure}

\newpage
\section{Contractibility of spaces of geodesic triangulations}

  The field of discrete differential geometry features the discretization of the whole theory of classical geometry and topology of surfaces, in contrast to traditional numerical methods focusing on estimating the approximation errors. Recently, important progress has been made in this direction, including the theory of discrete conformal maps, discrete harmonic maps, and discrete uniformization theorems \cite{gu2018discrete, gu2018discrete2,  gu2019convergence, luo2020discrete, sun2015discrete, wu2015rigidity, wu2014finiteness, wu2020convergence}. 
  
  The spaces of geodesic triangulations can be regarded as discretization of groups of surface diffeomorphisms. It is conjectured in \cite{connelly1983problems} that $\mathcal{GT}(S, T)$ is homotopy equivalent to the group of isometries of $S$ homotopic to the identity, when $S$ is equipped with a metric with constant curvature. This conjecture has been confirmed by \cite{luo2021deformation, luo2021deformation2, erickson2021planar} in the cases of hyperbolic surfaces and flat tori. Progress has been made in the case of spheres in special cases of Delaunay triangulations \cite{luo2022deformation3}, but the general case remains open. These results can be regarded as discrete versions of classical results Earle and Eells \cite{earle1969fibre} on homotopy types of surface diffeomorphisms.

In this section we show the contractibility of $\mathcal{GT}(\Omega, T)$ if $\Omega$ is a non-convex quadrilateral.

\subsection{Contractibility of spaces of geodesic triangulations of quadrilateral}
We recall the following definitions from \cite{luo2022spaces}. An embedded $n$-sided polygon $\Omega$ in the plane is determined by a map $\phi$ from the set $V_B = \{v_1,v_2, \cdots, v_n, v_{n+1} = v_1\}$ to $\mathbb{R}^2$, and line segments connecting the images under $\phi$ of two consecutive vertices in $V_B$ . We assume that $T = (V, E, F)$ is a triangulation of $\Omega$ with vertices $V$, edges $E$ and faces $F$ such that  $V = V_I\cup V_B$, where $V_I$ is the set of interior vertices and $V_B$ is the set of boundary vertices.

A \textit{geodesic triangulation} with combinatorial type  $T$ of $\Omega$ is an embedding $\psi$ of the $1$-skeleton of $T$ in the plane  such that $\psi$ agrees with $\phi$ on $V_B$, and $\psi$ maps every edge in $E$ to a line segment parametrized by arc length. The set of these maps is called the \textit{space of geodesic triangulations on $\Omega$ with combinatorial type $T$}, and is denoted by $\mathcal{GT}(\Omega, T)$. Each geodesic triangulation is uniquely determined by the positions of the interior vertices in $V_I$, so the topology of $\mathcal{GT}(\Omega, T)$ is induced by $\Omega^{|V_I|} \subset \mathbb{R}^{2|V_I|}$.

Let $P$ is a non-convex quadrilateral with one \emph{reflexive vertex} $v_0$, meaning that the inner angle at $v_0$ is larger than $\pi$. Let $T$ be a triangulation of $P$. Let $\tilde{P}$ denote the convex hull of $P$, which in this case is a triangle. Since  $\tilde{P}$ is convex,  $\mathcal{GT}(\tilde{P}, \tilde T)$ is a contractible space, where $v_0$ is regarded as an interior vertex  . 

The idea is that we can construct a fibration $\mathcal{GT}(\tilde{P}, \tilde T)$ to the interior of $\tilde{P}$ denoted by $\Delta$. This map is defined by sending a geodesic triangulation $\tau \in X(\tilde{P})$ to the position of $v_0 \in \Delta$.

$$\pi: X(\tilde{P}) \to \Delta$$
$$\pi(\tau) = \tau(v_0)$$

\begin{figure}[h!]
  \includegraphics[width=0.3\linewidth]{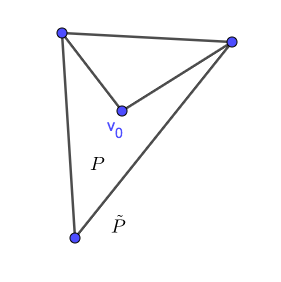}
\end{figure}

\begin{theorem}
 If $P$ is a quadrilateral with a triangulation $T$, then $\mathcal{GT}(P, T)$ is contractible. 
\end{theorem}

\begin{proof}
The statement is true if $P$ is a convex. If $P$ is not convex, we prove that the projection $\pi$ is a fibration, using the fact that there exists a unique projective transformation mapping a given non-convex quadrilateral to another non-convex quadrilateral. 

 For any given $v_0\in \tilde{P}$, $\pi^{-1}(v_0)$ is the space of geodesic triangulations of the quadrilateral with the reflexive vertex $v_0$, denoted as $X(P_{v_0})$. Fix an element $\tau \in \pi^{-1}(v_0)$. For any other element $v_0' \in \Delta$, the unique projective transformation $\phi$ mapping $P_{v_0}$ to $P_{v_0'}$ sends $\tau$ to an element $\phi(\tau)\in \pi^{-1}(v_0') =X(P_{v_0'})$. This projective transformation clearly depends continuously on $v_0'$. Hence this gives a global section $\sigma: \Delta \to X(\tilde{P})$. Since $X(\tilde{P})$ and $\Delta$ are contractible, the fibration is trivial and each fiber $\pi^{-1}(v_0)$ is contractible. 
\end{proof}

The contractibility problem for general star-shaped polygons remains unknown. Based on the evidence of Bing-Starbird \cite{bing1978linear}, we conjecture that 
\begin{conjecture}
If $P$ is a star-shaped polygon with a triangulation $T$ such that $T$ has no dividing edge, then $\mathcal{GT}(P, T)$ is contractible. 
\end{conjecture}

\bibliography{ref} 
\bibliographystyle{amsplain}

\end{document}